%% file: ahst.tex
\documentclass[a4paper,11pt]{article}

\usepackage{graphicx}
\usepackage{caption,color}
\usepackage{amsmath,amssymb,amsthm,enumerate}
\usepackage{fullpage,times}
\usepackage[algo2e,ruled]{algorithm2e}
\usepackage{verbatim}

\newtheorem{theorem}{Theorem}[section]
\newtheorem{lemma}[theorem]{Lemma}
\newtheorem{observation}[theorem]{Observation}

\newcommand{\newrestatedthm}[1]{\theoremstyle{plain}\newtheorem*{theorem-#1}{Theorem \ref{#1}}}

\newcommand{\junk}[1]{}
\newcommand{\ignore}[1]{}

\def\floor#1{\lfloor #1 \rfloor}
\def\ceil#1{\lceil #1 \rceil}

\newcounter{note}[section]

\title{The $(h,k)$-Server Problem on Bounded Depth Trees%
\footnote{This work was supported by NWO grant 639.022.211, ERC consolidator
grant 617951, and NCN grant DEC-2013/09/B/ST6/01538.  It was carried out while
{\L}.~Je{\.z} was a postdoc at TU/e.}
}
\author{Nikhil Bansal\thanks{TU Eindhoven, Netherlands.
\texttt{\{n.bansal,m.elias,g.koumoutsos\}@tue.nl}},
Marek Eli\'a\v{s}\footnotemark[2],
{\L}ukasz Je\.{z}\thanks{University of Wroc{\l}aw, Poland.
\texttt{lje@cs.uni.wroc.pl}},
Grigorios Koumoutsos\footnotemark[2]}

\begin{document}

\maketitle

\begin{abstract}

We study the $k$-server problem in the resource augmentation setting i.e.,~when the
performance of the online algorithm with $k$ servers is compared to the offline
optimal solution with $h \leq k$ servers. The problem is very poorly understood beyond uniform metrics.  
For this special case, the classic $k$-server algorithms are roughly $(1+1/\epsilon)$-competitive when $k=(1+\epsilon) h$,
for any $\epsilon >0$. Surprisingly however, no $o(h)$-competitive algorithm is known even for HSTs of depth 2 and even when $k/h$ is arbitrarily large. 

We obtain several new results for the problem. First we show that the known $k$-server algorithms do not work even on very simple metrics. In particular, the Double Coverage algorithm has competitive ratio $\Omega(h)$ irrespective of the value of $k$, even for depth-2 HSTs. Similarly the Work Function Algorithm, that is believed to be optimal for all metric spaces when $k=h$, has competitive ratio $\Omega(h)$ on depth-3 HSTs even if $k=2h$. 
Our main result is a new algorithm  that is $O(1)$-competitive for constant depth trees, whenever $k =(1+\epsilon )h$ for any $\epsilon > 0$. Finally, we give a general lower bound that any deterministic online algorithm has competitive ratio at least 2.4
even for depth-2 HSTs and when $k/h$ is arbitrarily large. This gives a surprising qualitative separation between uniform metrics and depth-2 HSTs for the $(h,k)$-server problem, and gives the strongest known lower bound for the problem on general metrics.
 


\end{abstract}

\setcounter{page}{0}
\newpage

\section{Introduction}
\input{newintro}


\input{techniques}

\input{lbs2}

\section{Algorithm for depth-$d$ trees}\label{sec:algo}
\input{algo}


\section{Concluding Remarks}
\input{conclusion}
\newpage
\bibliographystyle{plain}
\bibliography{references-k-server}

\appendix

\include{appendix}

\section{Lower Bounds}\label{sec:lbs}
\input{lbs1}
\input{lbs3}

\end{document}

%% file: newintro.tex
The classic $k$-server problem, introduced by Manasse et al.~\cite{MMS90},
is a broad generalization of various online problems and is defined as follows. 
There are $k$ servers that reside on some points of a given metric space. At each step,
a request arrives at some point of the metric space and must be served by moving
some server to that point. The goal is to minimize the total
distance traveled by the servers.

In this paper, we study the resource augmentation setting of the problem, also known as the ``weak adversary'' model~\cite{Kou99}, 
where the online algorithm has $k$ servers, but its performance is compared to a ``weak'' offline optimum with $h \leq k$ servers.
We will refer to this as the $(h,k)$-server problem. 
Our motivation is twofold. Typically, the resource augmentation setting gives a much more refined view of the problem and allows one to bypass 
overly pessimistic worst case bounds, see e.g.~\cite{KalPruhs00}. Second, as we discuss below, the $(h,k)$-server problem is much less understood than the $k$-server problem and seems much more intriguing.

\subsection{Previous work}
The $k$-server problem has been extensively studied;
here we will focus only on deterministic algorithms.
It is well-known that no algorithm can be better than $k$-competitive for any metric space on more than $k$ points~\cite{MMS90}.  
In their breakthrough result, Koutsoupias and Papadimitriou~\cite{KP95} showed
that the Work Function Algorithm (WFA) is $(2k-1)$-competitive in any metric
space. For special metrics such as the uniform metrics%
\footnote{The $k$-server problem in uniform metrics is equivalent to
the paging problem.}, the line, and trees,
tight $k$-competitive
algorithms are known (cf.~\cite{BEY98}).
It is widely believed that a $k$-competitive algorithm exists for every metric
space (the celebrated $k$-server conjecture), and it is also plausible
that the WFA achieves this guarantee. Qualitatively, this means that general
metrics are believed to be no harder than the simplest possible case of uniform
metrics.

\vspace{2mm}
\noindent{\bf The $(h,k)$-server problem.}
Much less is known for the $(h,k)$-server problem.
In their seminal paper~\cite{ST85}, Sleator and Tarjan
gave several $\frac{k}{k-h+1}$-competitive algorithms for the uniform metrics
and also showed that this is the best possible ratio.
This bound was later extended to the weighted
star metric (weighted paging)~\cite{Young94}.
Note that this guarantee equals $k$ for $k=h$ (the usual $k$-server setting),
and tends to $1$ as $k/h$ approaches infinity.
In particular, for $k=2h$, this is smaller than $2$.

It might seem natural to conjecture that, analogously to the $k$-server case,
general metrics are no harder than the uniform metrics,
and hence that $k/(k-h+1)$
is the right bound for the $(h,k)$-server problem in all metrics.
However, surprisingly, Bar-Noy and Schieber (cf.~\cite[p.~175]{BEY98})
showed this to be false:
In the line metric, for $h=2$,
no deterministic algorithm can be better than 2-competitive, regardless of the
value of $k$.
This is the best known lower bound for the general $(h,k)$-server problem.

On the other hand, the best known upper bound is $2h$,
even when $k/h \rightarrow \infty$.
In particular, Koutsoupias~\cite{Kou99} showed that the WFA with $k$-servers
is (about) $2h$-competitive against an offline optimum with $h$ servers.
Note that one way to achieve a guarantee of $2h-1$ is simply to disable the
$k-h$ extra online servers and use WFA with $h$ servers only.
The interesting thing about the result of \cite{Kou99} is that
the online algorithm does not know $h$ and is $2h$-competitive
simultaneously for every $h \leq k$.
But, even if we ignore this issue of whether the online algorithm knows $h$ or
not, no guarantee better than $h$ is known, even for very special metrics such
as depth-$2$ HSTs or the line, and even when $k/h \rightarrow \infty$.

\subsection{Our Results}
Motivated by the huge gap between the known lower and upper bounds even for very
simple metrics, we consider bounded-depth HSTs (defined formally
in Section \ref{sec:prelim}).

We first show very strong lower bounds on all the previously known algorithms
(beyond uniform metrics), specifically the Double Coverage (DC) algorithm of Chrobak et al.
\cite{CKPV91,CL91} and the WFA.
This is perhaps surprising because, for the $k$-server problem,
DC is optimal in trees and WFA is believed to be optimal in all metrics.

\begin{theorem}\label{thm:lb-dc}
The competitive ratio of DC in depth-$2$ HSTs is
$\Omega(h)$, even when $k/h \rightarrow \infty$.
\end{theorem}

In particular, DC is unable to use the extra servers in a useful way.
A similar lower bound was recently shown for the line by a superset of the
authors \cite{BEJKP15}.

For the WFA, we present the following lower bound.
\begin{theorem}\label{thm:lb-wfa}
The competitive ratio of the WFA is at least
$h+1/3$ in a depth-$3$ HST for $k=2h$. 
\end{theorem}

Surprisingly, the lower bound here is strictly larger than $h$! Recall that the WFA
is believed to be $h$-competitive for $k=h$.
Although the lower bound instance is quite simple, the analysis is rather
involved as we need to consider how the various work function values
evolve over time.
The lower bound in Theorem \ref{thm:lb-wfa} can also be extended to
the line metric.
Interestingly, it exactly matches the upper bound $(h+1) \mbox{OPT}_h - \mbox{OPT}_k $ implied by the result of
Koutsoupias \cite{Kou99} for the WFA in the line.
We describe the details in Appendix~\ref{sec:lb-wfa-ln}.

Our main result is the first $o(h)$-competitive algorithm
for depth-$d$ trees with the following guarantee.

\begin{theorem}\label{thm:alg-ub}
There is an algorithm that is $O_d(1)$-competitive on any depth-$d$ tree,
whenever $k = \delta h$ for $\delta > 1$.
More precisely, its competitive ratio is
$O(d\cdot (\frac{\delta^{1/d}}{\delta^{1/d}-1})^{d+1})$.
If $\delta = 1+\epsilon$, for $0<\epsilon\leq 1$,
this equals to $O(d\cdot (2d/\epsilon)^{d+1})$,
and improves to $O(d \cdot 2^{d+1})$ for $\delta \geq 2^d$.
\end{theorem}

The algorithm is designed to overcome the drawbacks of DC and WFA, and can be
viewed as a more aggressive and cost-sensitive version of DC.
It moves the servers more aggressively at non-uniform speeds towards the
region of the current request,
giving a higher speed to a server located in a region containing many servers.
It does not require the knowledge of $h$,
and is simultaneously competitive against all $h$ strictly smaller than $k$.


Finally, we give an improved general lower bound.
Bar-Noy and Schieber (cf.~\cite[p.~175]{BEY98}) showed that there is no better
than $2$-competitive algorithm for the $(h,k)$-server problem in general
metrics, by constructing their lower bound in the line metric.
Our next result shows that even a $2$-competitive algorithm is not possible.
In particular, we present a construction in a depth-$2$ HST showing that
no $2.4$-competitive algorithm is possible.

\begin{theorem}\label{thm:gen_lb}
There is no $2.4$-competitive deterministic algorithm for general metrics,
even when $k/h \rightarrow \infty$, provided that $h$ is larger than some
constant independent of $k$.
\end{theorem}

This shows that depth-2 HSTs are qualitatively quite different from depth-1 HSTs
(same as uniform metrics) which allow a ratio $k/(k-h+1)$.
We have not tried to optimize the constant 2.4 above, but computer experiments
suggest that the bound can be improved to about $2.88$.

%% file: techniques.tex
\subsection{Notation and Preliminaries}\label{sec:prelim}


In the $(h,k)$-setting, we define the competitive ratio as follows.
An online algorithm $ALG$ is $R$-competitive in metric $M$,
if
$ALG_k(I) \leq R\cdot OPT_h(I) + \alpha$ holds
for any finite request sequence $I$ of points in $M$.
Here,
$ALG_k(I)$ denotes the cost of $ALG$ serving $I$ with $k$ servers,
$OPT_h(I)$ denotes the optimal cost to serve $I$ with $h$ servers,
and $\alpha$ is a constant independent of $I$.
An excellent reference on competitive analysis is \cite{BEY98}.

A {\em depth-$d$ tree} is an edge-weighted rooted tree
with each leaf at depth exactly $d$.
In the $(h,k)$-server problem in a depth-$d$ tree,
the requests arrive only at leaves,
and the distance between two leaves is defined as
the distance in the underlying tree.
A depth-$d$ HST is a depth-$d$ tree with the additional property that the distances decrease geometrically away from the root (see e.g.~Bartal~\cite{Bartal96}).
We will first present our algorithm for general depth-$d$ trees
(without the HST requirement), and later show how to (easily) extend it to 
arbitrary trees with bounded diameter, where requests are also allowed in the
internal nodes.

\medskip
\noindent{\bf DC Algorithm for Trees.}
Our algorithm can be viewed as a non-uniform version of the DC algorithm
for trees \cite{CKPV91,CL91}, which works as follows.
When a request arrives at a vertex $v$ of the tree, all the servers
adjacent to $v$ move towards it along the edges at the same speed until one
eventually arrives at $v$. Here, a server $s$ is \textit{adjacent} to a
point $x$ if there is no other server on the (unique) path from $s$
to $x$. If multiple servers are located at a single point,
only one of them is chosen arbitrarily.

\medskip
\noindent{\bf Work Function Algorithm.}
Consider a request sequence $\sigma = r_1, r_2,\dotsc, r_m$. For each
$i=1,\dotsc,m$, let $w_i(A)$ denote the optimal cost to serve requests
$r_1, r_2,\dotsc, r_i$ and end up in the configuration $A$,
which is specified by the set of locations of the servers.
The function $w_i$ is called \textit{work function}.
The Work Function Algorithm (WFA) decides its moves depending on the values of
the work function.
Specifically, if the algorithm is in a configuration $A$ and a request $r_i
\notin A$ arrives, it moves to the configuration $X$
such that $r_i \in X$ and $w_i(X) + d(A,X)$ is minimized. For
more background on the WFA, see~\cite{BEY98}.

\subsection{Organization}
In Section \ref{sec:lb_dc}, we describe the lower bound for the DC
in depth-$2$ HSTs.
The shortcomings of the DC might help the reader
to understand the motivation behind the design of
our algorithm for depth-$d$ trees, which we describe in Section \ref{sec:algo}.
Its extension to the bounded-diameter trees can be found in
Appendix~\ref{sec:alg_diam}.
Due to the lack of space, the lower bound for the WFA
(Theorem~\ref{thm:lb-wfa}) and the general lower bound
(Theorem~\ref{thm:gen_lb})
are also located in the appendix.

%% file: lbs2.tex
\section{Lower Bound for the DC Algorithm on depth-$2$ HSTs}\label{sec:lb_dc}
We now show a lower bound of $\Omega(h)$ on the competitive ratio of the DC algorithm. 

Let $T$ be a depth-2 HST with $k+1$ subtrees and edge lengths chosen as
follows. Edges from the root $r$ to its children have length $1-\epsilon$, and
edges from the leaves to their parents length $\epsilon$ for some $\epsilon \ll
1$. Let $T_u$ be a subtree rooted at an internal node $u \neq r$.
A {\em branch} $B_u$ is defined as $T_u$ together with the edge $e$ connecting
$T_u$ to the root. We call $B_u$ {\em empty}, if there is no online
server in $T_u$ nor in the interior of $e$. Since $T$ contains $k+1$ branches,
at least one of them is always empty.

The idea behind the lower bound is quite simple.
The adversary moves all its $h$ servers to the leaves of an empty
branch $B_u$, and keeps requesting those leaves until DC brings $h$ servers to
$T_u$.
Then, another branch has to become empty,
and the adversary moves all its servers there, starting
a new {\em phase}.
The adversary can execute an arbitrary number of such phases.

The key observation is that DC is ``too slow'' when bringing new servers to
$T_u$, and incurs a cost of order $\Omega(h^2)$ during each phase,
while the adversary only pays $O(h)$.

\newrestatedthm{thm:lb-dc}
\begin{theorem-thm:lb-dc}
The competitive ratio of DC in depth-$2$ HSTs is
$\Omega(h)$, even when $k/h \rightarrow \infty$.
\end{theorem-thm:lb-dc}

\begin{proof}   
We describe a phase, which can be repeated arbitrarily many times.
The adversary places all its $h$ servers at different
leaves of an empty branch $B_u$ and does not move until the end of the phase.
At each time during the phase, a request arrives at such a leaf,
which is occupied by some offline server, but contains no online servers.
The phase ends at the moment when the $h$th server of DC arrives to $T_u$.

Let $ALG$ denote the cost of the DC algorithm and $ADV$ the cost of the
adversary during the phase.
Clearly, $ADV = 2h$ in each phase: The adversary moves
its $h$ servers to $T_u$ and does not incur any additional cost until the end of
the phase.
However, we claim that $ALG = \Omega(h^2)$,
no matter where exactly the DC servers are located when the phase starts.
To see that, let us call Step $i$ the part of the phase when DC has exactly
$i-1$ servers in $T_u$.
Clearly, Step 1 consists of only a single request, which
causes DC to bring one server to the requested leaf.
So the cost of DC for Step 1 is at least $1$.
To bound the cost in the subsequent steps, we make the following observation.

\begin{figure}[t!]
\hfill\includegraphics{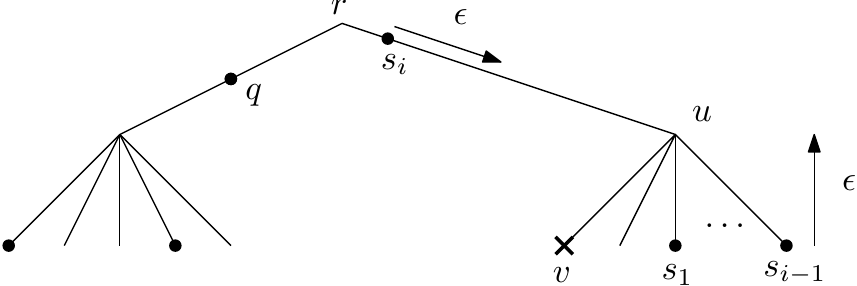}\hfill\ %
\caption{Move of DC during Step $i$.
Servers $s_1, \dotsc, s_{i-1}$ are moving towards $u$ by distance $\epsilon$
and $s_i$ is moving down the edge $(r,u)$ by the same distance.
While $s_i$ is in the interior of $(r,u)$, no server $q$
from some other branch is adjacent to $v$ because the unique path
between $v$ and $q$ passes through $s_i$.}
\label{fig:lb-dc}
\end{figure}

\begin{observation}
At the moment when a new server $s$ enters the subtree $T_u$,
no other DC servers are located along the edge $e=(r,u)$.
\end{observation}
This follows from the construction of DC, which moves only servers adjacent to
the request. At the moment when $s$ enters the edge $e$, no other server above
$s$ can be inside $e$; see Figure \ref{fig:lb-dc}.

We now focus on Step $i$ for $2 \leq i \leq h$.
There are already $i-1$ servers in $T_u$, and
let $s_i$ be the next one which is to arrive to $T_u$.

Crucially, $s_i$ moves if and only if all the servers
of DC in the subtree $T_u$ move from the leaves towards $u$,
like in Figure \ref{fig:lb-dc}:
When the request arrives at $v$, $s_i$ moves by $\epsilon$ and
the servers inside $T_u$ pay together $(i-1)\epsilon$.
However, such a moment does not occur, until all servers
$s_1, \dotsc, s_{i-1}$ are again at leaves, i.e. they incur an additional cost $(i-1)\epsilon$.
To sum up, while $s_i$ moves by $\epsilon$, the servers inside $T_u$ incur cost
$2(i-1)\epsilon$.

When Step $i$ starts, the distance of $s_i$ from $u$ is at least $1-\epsilon$,
and therefore $s_i$ moves by distance $\epsilon$ at least
$\lfloor \frac{1 - \epsilon}{\epsilon} \rfloor$ times,
before it enters $T_u$.
So, during Step $i$, DC pays at least
\[
\left\lfloor \frac{1 - \epsilon}{\epsilon} \right\rfloor(2(i-1)\epsilon
	+ \epsilon)
	\geq \frac{1 -2 \epsilon}{\epsilon} \cdot \epsilon \big(2(i-1)+1\big)
	= (1-2\epsilon)(2i-1)
\]

By summing over all steps $i = 1, \dotsc, h$ and choosing $\epsilon \leq 1/4$,
we get
\[
ALG \geq 1 +  \sum_{i=2}^h (1-2\epsilon)(2i-1)
	\geq \sum_{i=1}^h (1-2\epsilon)(2i-1)
	= (1-2\epsilon) h^2 \geq \frac{h^2}{2}
\]
To conclude the proof, we note that
$ALG/ADV \geq (h^2/2)/(2h) = h/4 = \Omega(h)$
for all phases.
\end{proof}

%% file: algo.tex

In this section we prove Theorem \ref{thm:alg-ub}.

Recall that a depth-$d$ tree is a rooted-tree and we allow the requests to appear only at the leaves.
However, to simplify the algorithm description, we will allow the online
servers to reside at any node or at any location on an edge (similar to that in
DC). To serve a request at a leaf $v$, the algorithm moves all the adjacent
servers slowly
towards $v$, where the speed of each server coming from a different branch
of the tree depends on the number of the online servers ``below" it. 

To describe the algorithm formally, we
state the following definitions.
For a point $x \in T$ (either a node or a location on some edge),
we define $T_x$ as the subtree consisting of all points below
$x$ including $x$ itself, and we denote $k_x$ the number of the online servers
inside $T_x$. 
If $s$ is a server located at a point $x$, we denote $T_s=T_x$ and $k_s=k_x$.
We also denote $T_x^- = T_x \setminus \{x\}$, and $k_x^-$
the corresponding number of the algorithm's servers in $T_x^-$.
If $u$ is a level-1 node, we call $T_u$ an {\em elementary} subtree, see
Figure~\ref{fig:ub_phase2}.

\subsection{Algorithm Description}
Suppose a request arrives at a leaf $v$ that lies in the elementary subtree
$T_u$.
The algorithm  proceeds in two phases, depending on whether there is a server
along the path from $v$ to the root $r$ or not. 
We set speeds as described in Algorithm~\ref{alg:ub} below and move the
servers towards $v$ either until the phase ends or the set $A$ of servers
adjacent to $v$ changes. This defines the elementary moves (where $A$ stays unchanged).
Note that if there are some servers in the path between $v$ and the root $r$,
only the lowest of them belongs to $A$. We denote this server $q$ during an elementary move.
Figure \ref{fig:ub_phase2} shows an elementary subtree $T_u$ and the progress of
Phase~2.

\begin{algorithm2e}
\DontPrintSemicolon
\textbf{Phase 1:} While there is  no server along the path $r-v$\;
\Indp
For each $s \in T_u$: move $s$ at speed $1/k_u$\;
For each $s \in A \setminus T_u$: move $s$ at speed $k_s/(k-k_u)$\;
\BlankLine
\Indm
\textbf{Phase 2:} While no server reached $v$;
	Server $q \in A$ moves down along the path $r-v$\;
\Indp
For server $q$: move it at speed 1\;
For each $s\in A \setminus \{q\}$: move it at speed $k_s/k_q^-$\;
\caption{Serving request at leaf $v$ in elementary subtree $T_u$.}
\label{alg:ub}
\end{algorithm2e}

\begin{figure}
\hfill\includegraphics{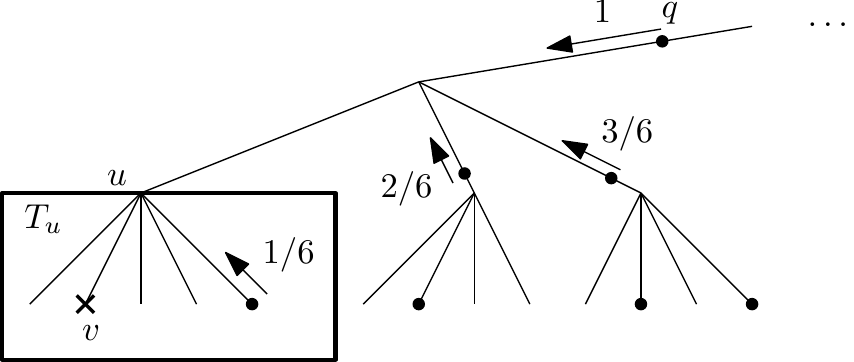}\hfill\ %
\caption{A request at $v$ inside of the elementary tree $T_u$, and Phase~2 of
Algorithm \ref{alg:ub}. Note that $k_q^-$ equals $6$ in the visualised
case. Speed is noted next to each server moving.}
\label{fig:ub_phase2}
\end{figure}

We note two properties, the first of which follows directly from the design of
Algorithm~\ref{alg:ub}.

\begin{observation}\label{obs:ub_s_in_e}
	No edge $e\in T$ contains more than one server of Algorithm~\ref{alg:ub} in its
	interior.
\end{observation}

Note that during both the phases, the various servers move at non-uniform speeds depending on their respective $k_s$.
The following observations about these speeds will be useful.

\begin{observation}\label{obs:ub_speed}
During Phase 1, the total speed of servers inside $T_u$ is 1, unless there
are no servers inside $T_u$.
This follows as there are $k_u$ servers moving at speed $1/k_u$.
Similarly, the total speed of servers outside $T_u$ (if there
are any) is also 1. This follows as $\sum_{s \in A \setminus T_u} k_s = k-k_u$.

Analogously, for Phase 2: the total speed of servers
inside $T_q^-$ is 1, if there are any.
This follows as $\sum_{s \in A \setminus \{q\}} k_s = k_q^{-}$.
\end{observation}

The intuition behind the algorithm is the following. Recall that the problem with DC is that it brings the outside
servers too slowly into $T_u$ when requests start arriving there.
Therefore, our algorithm changes the speeds of the servers adjacent to $v$ to make the algorithm more
aggressive. The servers outside $T_u$ in Phase~1 and the server $q$ in Phase~2
are {\em helpers} coming to aid the servers inside $T_u$ and
$T_q^-$. From each region, these helpers move at the speed proportional to the number of
the online servers in that region, which helps in removing them quickly from regions with excess servers.
Of course, the online algorithm does not know which regions have excess servers and which ones do not (as it does not know the offline state), but the number of servers in a region serves as a good measure of this.
The second main idea is to keep the total speed of the helpers proportional to the total speed of
the servers inside $T_u$ and $T_q^-$. This prevents the algorithm from becoming
overly aggressive and keeps the cost of the moving helpers comparable to the cost incurred within $T_u$ and $T_q^{-1}$.
As for DC, the analysis of our algorithm is based on a carefully designed potential function.

\subsection{Analysis}
We will analyze the algorithm based on a suitable potential function $\Phi(t)$.
Let $\mbox{ALG}(t)$ and $\mbox{OPT}(t)$ denote the cost of the algorithm and of the adversary respectively, for serving the 
request at time $t$.
Let
$\Delta_t \Phi = \Phi(t) - \Phi(t-1)$ denote the change of the potential at 
time $t$.
We will ensure that $\Phi$ is non-negative and bounded from above by
a function of $h, k, d$, and length of the longest edge in $T$.
Therefore, to show $R$ competitiveness, it suffices to show
that the following holds at each time $t$:
$ \mbox{ALG}(t) + \Delta_t \Phi \leq R \cdot \mbox{OPT}(t).$

To show this, we split the analysis into two parts:
First, we let the adversary move his
servers to serve the request. Then we consider the move of the algorithm.
Let $\Delta^{OPT}_t \Phi$ and $\Delta^{ALG}_t \Phi$ denote the changes in $\Phi$
due to the move of the adversary and the algorithm respectively. Clearly,
$\Delta_t \Phi = \Delta^{OPT}_t \Phi + \Delta^{ALG}_t \Phi$, and thus it suffices to show the following two inequalities:
\begin{align}
\label{eq:ub_opt}
\Delta^{OPT}_t \Phi &\leq  R \cdot OPT(t)\\
\label{eq:ub_alg}
ALG(t) + \Delta^{ALG}_t \Phi & \leq 0
\end{align}

%

\subsubsection{Potential function}
Before we define the potential, we need to formalize the notion of excess and
deficiency in the subtrees of $T$.
Let $d(a,b)$ denote the distance of points $a$ and $b$.
For $e=(u,v) \in T$, where $v$ is the node closer to the root,
we define $k_e := k_u + \frac{1}{d(u,v)} \sum_{s\in e} d(s,v)$.
Note that this is the number of online servers in $T_u$, plus the possible 
single server in $e$ counted fractionally, proportionally to its position along 
$e$. For an edge $e$, let $\ell(e)$ denote its level with the convention that
the edges from leaf to their parents have
level $1$, and the edges from root to its children have level $d$.
For $\ell=1,\ldots,d$, let $\beta_{\ell}$ be some geometrically increasing constants that will be defined later.
For and edge $e$ we define the \emph{excess $E_e$ of $e$} and the
\emph{deficiency $D_e$ of $e$} as follows
\begin{align*}
E_e &= \max \{k_e - \lfloor \beta_{\ell(e)} \cdot h_u\rfloor, 0\} \cdot d(u,v) &
D_e &= \max \{\lfloor\beta_{\ell(e)} \cdot h_u\rfloor - k_e, 0\} \cdot d(u,v).
\end{align*} 
Note that these compare $k_e$ to $h_u$ with respect to the \emph{excess threshold} $\beta_{\ell(e)}$.
We call an edge {\em excessive}, if $E_e > 0$,
otherwise we call it {\em deficient}.
Let us state a few basic properties of these two terms.

\medskip
\begin{observation}
Let $e$ be an edge containing an algorithm's server $s$ in its interior.
If $e$ is excessive, it cannot become deficient unless $s$ moves upwards
completely outside of the interior of $e$.
Similarly, if $e$ is deficient, it cannot become excessive unless $s$
leaves interior of $e$ completely.
\end{observation}
\noindent
Note that no other server can pass through $e$ while $s$ still resides there
and the contribution of $s$ to $k_e$ is a nonzero value strictly smaller than 1,
while $\floor{\beta_\ell h_u}$ is an integer.

\medskip
\begin{observation}\label{obs:alg_move}
Let $e$ be an edge and $s$ be an algorithm's server in its interior moving by
a distance $x$. Then either $D_e$ or $E_e$ changes exactly by $x$.
\end{observation}
\noindent
This is because $k_e$ changes by $x/d(u,v)$, and therefore the
change of $D_e$ (resp.~$E_e$) is $x$.

\medskip
\begin{observation}\label{obs:adv_move}
If an adversary server passes through the edge $e$, change in
$D_e$ (resp.~$E_e$) will be at most $\ceil{\beta_\ell} \cdot d(u,v)$.
\end{observation}
\noindent
To see this, note that
$\floor{\beta_{\ell(e)} h_u}
	\leq \floor{\beta_{\ell(e)} (h_u-1)} + \ceil{\beta_{\ell(e)}}$.

\newcommand{\AlphaEd}{\frac{\delta}{\delta-\beta_d}
        \left( 3 + \frac{\beta_d}{\delta} \alpha^D_d \right)}%
\newcommand{\BetaL}{2^{\ell-1}}%
\newcommand{\AlphaDL}{2\ell-1}%
\newcommand{\AlphaEL}{%
	\sum_{i=\ell}^{d-1} \gamma^{i-\ell+1}
		\left(2 + \frac1\beta \alpha^D_i\right)
	+ \gamma^{d-\ell} \alpha^E_d
	}%
\newcommand{\Ratio}{2^d(d+4)}%

\medskip 

We now fix the excess thresholds.
We first define $\beta$ depending on $\delta = k/h$ as, 
\[\beta = 2 \textrm{ if } \delta \geq 2^d,
\textrm{ and } \beta = \delta^{1/d} \textrm{ for }\delta \leq 2^d.\]
For each $\ell = 1, \dotsc, d$, we define the excess threshold for all edges in
the level $\ell$ as $\beta_\ell := \beta^{\ell-1}$.
We also denote $\gamma := \frac{\beta}{\beta-1}$.
Note that, for all possible $\delta > 1$, our choices satisfy
$1 < \beta \leq 2$ and $\gamma \geq 2$.
Now, we can define the potential. Let
\[ \Phi := \sum_{e\in T} \left( \alpha^D_{\ell(e)} D_e + \alpha^E_{\ell(e)}E_e\right),
\]
where the coefficients $\alpha^D_\ell$ and $\alpha^E_\ell$ are as follows:
\begin{align*}
\alpha^D_\ell &:= \AlphaDL && \text{for $\ell = 1,\dotsc, d$}\\
\alpha^E_d &:= \AlphaEd \qquad \\
\alpha^E_\ell &:= \AlphaEL && \text{for $\ell = 1,\dotsc, d-1$}.
\end{align*}
Note that $\alpha^D_\ell > \alpha^D_{\ell-1}$
and $\alpha^E_\ell < \alpha^E_{\ell-1}$ for all $1 < \ell \leq d$. The latter follows as the multipliers $\gamma^{i-\ell+1}$ and $\gamma^{d-\ell}$ decrease with increasing $\ell$ and moreover the summation in the first term of $\alpha^E_\ell$ has fewer terms as $\ell$ increases.

To prove the desired competitive ratio for Algorithm \ref{alg:ub},
the idea will be to show that the {\em good} moves (when a server enters a
region with deficiency, or leaves a region with excess) contribute more than
the {\em bad} moves (when a server enters a region with excess, or leaves a
region that is already deficient).

As the dynamics of the servers can be decomposed into elementary moves, it suffices to only analyze these.
We will also assume that no servers of $A$ are located at a node.
This is without loss of generality, as only the moving servers can cause a change
in the potential, and each server appears at a node just for an
infinitesimal moment during its motion.

The following two lemmas give some properties of the deficient and excessive
subtrees, which will be used later in the proof of
Theorem~\ref{thm:alg-ub}. The proofs of both these Lemmas are located in
Appendix~\ref{sec:alg_omit}.

\begin{lemma}[Excess in Phase 1]\label{lem:ex_ph1}
Let us assume that no server of $A' = A\setminus T_u$ resides at a node.
For the set $E = \{ s \in A' \mid s\in e, E_e > 0\}$ of servers
which are located in the excessive edges,
and for $D = A' \setminus E$, the following holds.
If $k_u \leq \floor{\beta_2 h_u}$, then we have
\begin{align*}
\sum_{s\in D} k_s &\leq \beta_d (h-h_u) &
\text{ and }&&
\sum_{s\in E} k_s &\geq k - \beta_d h.
\end{align*}
\end{lemma}

\begin{lemma}[Excess in Phase 2]\label{lem:ex_ph2}
Let us assume that no server of $A' = A\setminus \{q\}$ resides at a node.
For the set $E = \{ s \in A' \mid s\in e, E_e > 0\}$ of
servers which are located in the excessive edges,
and for $D = A' \setminus E$, the following holds.
If $k_q^- \geq \floor{\beta_{\ell} h_q^-}$, then we have
\begin{align*}
\sum_{s\in D} k_s &\leq \frac1\beta k_q^- &
\text{ and }&&
\sum_{s\in E} k_s &\geq \frac1\gamma k_q^-.
\end{align*}
\end{lemma}

\subsubsection{Proof of Theorem \ref{thm:alg-ub}}

We now show the main technical result, which directly implies Theorem \ref{thm:alg-ub}.

\begin{theorem}
\label{thm:alg-ub2}
The competitive ratio of Algorithm \ref{alg:ub} in depth-$d$ trees
is $O(d\cdot \gamma^{d+1})$.
\end{theorem}

This implies Theorem \ref{thm:alg-ub} as follows.
If $\delta \geq 2^d$, we have $\beta = 2$ and
$\gamma = 2$, and we get the competitive ratio $O(d \cdot 2^{d+1})$.
For $1<\delta < 2^d$, we have $\beta = \delta^{1/d}$
and therefore the competitive ratio is
$O(d\cdot (\frac{\delta^{1/d}}{\delta^{1/d} - 1})^{d+1})$.
Moreover if $\delta = (1+\epsilon)$ for some $0<\epsilon \leq 1$, we have
$\beta = (1+\epsilon)^{1/d} \geq 1+\frac{1}{2d}$.
In this case, we get the ratio
$O(d \cdot (\frac{2d}{\epsilon})^{d+1})$,
as $\gamma \leq (1+\epsilon)^{1/d}\cdot \frac{2d}{\epsilon}$.
%
%

\smallskip
 
We now prove Theorem \ref{thm:alg-ub2}.

\medskip

\begin{proof}[Proof of Theorem \ref{thm:alg-ub2}]
As $\Phi$ is non-negative and bounded from above by
a function of $h, k, d$, and the length of the longest edge in $T$,
it suffices to show the inequalities \eqref{eq:ub_opt} and \eqref{eq:ub_alg}.


We start with \eqref{eq:ub_opt} which is straightforward.
By Observation~\ref{obs:adv_move}, the move of a single adversary's 
server through an edge $e$ of length $x_e$  changes $D_e$ or $E_e$
in the potential by at most $\ceil{\beta_{\ell(e)}} x_e$.
As the adversary incurs cost $x_e$ during this move,
we need to show the following inequalities:
\begin{align*}
\ceil{\beta_\ell} x_e \cdot \alpha^D_\ell &\leq R \cdot x_e
	\quad \text{for all } 1 \leq \ell \leq d\\ 
\ceil{\beta_\ell} x_e \cdot \alpha^E_\ell &\leq R \cdot x_e
	\quad \text{for all } 1 \leq \ell \leq d. 
\end{align*}
As we show in Lemma \ref{lem:alg_ratio} below,
$\ceil{\beta_\ell} \alpha^D_\ell$ and $\ceil{\beta_\ell} \alpha^E_\ell$
are of order $O(d\cdot \gamma^{d+1})$.
Therefore, there is $R = \Theta(d \cdot \gamma^{d+1})$, which satisfies
\eqref{eq:ub_opt}.

We now consider \eqref{eq:ub_alg} which is much more challenging to show.
Let us denote $A_E$ the set of edges containing some server from $A$ in their
interior.
We call an {\em elementary step} a part of the motion of the algorithm during
which $A$ and $A_E$ remain unchanged, and all the servers of $A$ are located in
the interior of the edges of $T$. 
Lemmas \ref{lem:ub_phase1} and \ref{lem:ub_phase2} below show that \eqref{eq:ub_alg} holds during an
elementary step, and the theorem would follow by summing \eqref{eq:ub_alg} over
all the elementary steps.
\end{proof}

\begin{lemma}\label{lem:ub_phase1}
During an elementary step in Phase 1 of Algorithm \ref{alg:ub},
the inequality \eqref{eq:ub_alg} holds.
\end{lemma}
\begin{proof}
Without loss of generality, let us assume that the elementary step lasted
exactly 1 unit of time.
This makes the distance traveled by each server equals to its speed, and makes calculations cleaner.

Let us first note that the cost $\mbox{ALG}$ incurred by the algorithm during this step
is at most $2$. Indeed, by Observation \ref{obs:ub_speed},
the total speed of the  servers in $T_u$ as well as in  $A\setminus T_u$
is $1$ (or $0$ if there are none), and thus the total speed of all servers is at most $2$.

To estimate the change of the potential $\Delta \Phi$, we decompose $A$ into four sets:
$D, \bar{D}$ are the servers outside (resp.~inside) $T_u$ residing in the 
deficient
edges, and $E, \bar{E}$ are the servers outside (resp.~inside) $T_u$ residing in
the excessive edges. Next, we evaluate $\Delta \Phi$ due to the
movement of the servers from each class separately.
We distinguish two cases:
\begin{enumerate}
\item {\em When $k_u \geq \floor{\beta_2 h_u} = \floor{\beta h_u}$.} The
servers from $\bar{D}$ are exactly those which have an offline server located
below them and their rise contributes to the increase of deficiency in their
edges. As one server of the adversary is already located at the requested
point where no online server resides, we have $|\bar{D}| \leq h_u - 1$.
Movement of the servers from $\bar{E}$ causes a decrease of the excess
in their edges
and we have $|\bar{E}|\geq k_u - |\bar{D}| = k_u - h_u + 1$.
Since all the servers in $T_u$ have speed $1/k_u$, the change of the potential
due to their movement is at most
$\frac{h_u-1}{k_u} \alpha^D_1 - (1 - \frac{h_u-1}{k_u}) \alpha^E_1$.
So, by our assumption,
$k_u \geq \floor{\beta h_u} \geq \beta h_u - 1 \geq \beta (h_u-1)$, and thus $(h_u-1)/k_u \leq 1/\beta$.

As for the servers outside $T_u$, all of them might be contained in $D$ in the
worst case, but they can cause an increase of $\Phi$ by at most $\alpha^D_d$.
Summing this up and using $\frac1\gamma = 1 - \frac1\beta$, we have
\[ \mbox{ALG} + \Delta\Phi
\leq 2
	+ \left(  \frac{h_u-1}{k_u} \right) \alpha^D_1
	- \left(1 - \frac{h_u-1}{k_u}\right) \alpha^E_1
	+ \alpha^D_d
\leq 2
	+ \frac1\beta \alpha^D_1
	- \frac1\gamma \alpha^E_1
	+ \alpha^D_d \leq 0,
\]
where the last inequality holds for the following reason:
We have
$\alpha^E_1 / \gamma \geq 2 + \frac1\beta \alpha^D_1 + \gamma^{d-2} \alpha^E_d$
(for $d\geq 2$), which cancels $2 + \frac1\beta \alpha^D_1$.
Also $\gamma^{d-2} \alpha^E_d$ cancels $\alpha^D_d$,
since $\gamma \geq 2$,
and thereby
$\gamma^{d-2} \alpha^E_d \geq 2^{d-2} \cdot 3 \geq 2d-1 = \alpha^D_d$.

\item {\em When $k_u < \floor{\beta_2 h_u}$.}
Here we rely on the decrease of the excess outside $T_u$, as
all the movement inside $T_u$ might contribute to the deficiency increase,
causing $\Phi$ to increase by $\alpha^D_1$.
As for the servers in $D$, their contribution to the increase of $\Phi$ satisfies
$\sum_{s\in D} k_s/(k-k_u) \alpha^D_{\ell(s)}
	\leq (\beta_d/\delta) \alpha^D_d$, as
\[ \frac{\sum_{s\in D} k_s}{k-k_u} \leq \frac{\beta_d(h-h_u)}{k-k_u}
	\leq \frac{\beta_d(h-h_u)}{\delta (h - k_u/\delta)},\text{ and }
\frac{h-h_u}{h-k_u/\delta} < 1.
\]
The first inequality follows from Lemma \ref{lem:ex_ph1}.
Fortunately, the movement of the servers in $E$ assures the desired decrease
of the potential as, by Lemma \ref{lem:ex_ph1}, we have
$\sum_{s\in E} k_s \geq k-\beta_d h$.
Thus, the movement of servers from $E$ causes the decrease of $\Phi$
by at least
$\sum_{s\in E} k_s/(k-k_u) \alpha^E_{\ell(s)}
        \geq ((\delta-\beta_d)/\delta) \alpha^E_d$, as
\[ \frac{\sum_{s\in E} k_s}{k-k_u}
	\geq \frac{k-\beta_d h}{k}
	\geq \frac{h(\delta-\beta_d)}{h \delta}.
\]
Summing everything up, we obtain
\[ ALG + \Delta\Phi \leq 2
	+ \alpha^D_1
	+ \frac{\beta_d}{\delta} \alpha^D_d
	- \frac{\delta - \beta_d}{\delta} \alpha^E_d
	\leq 0,
\]
where the second inequality holds by our choice of
$\alpha^D_1 = 1$ and $\alpha^E_d = \AlphaEd$.
\end{enumerate}
\end{proof}

\begin{lemma}\label{lem:ub_phase2}
During an elementary step in Phase 2 of Algorithm \ref{alg:ub},
the inequality \eqref{eq:ub_alg} holds.
\end{lemma}
\begin{proof}
Similarly to the proof of the preceding lemma, we
assume (without loss of generality) that the duration of the elementary step is
exactly $1$ time unit, so that the speed of each server equals the distance it
travels.
As the speed of the server $q$ is $1$, it also moves by distance 1.
The servers in $T_{q}^-$ (strictly below $q$)
move in total by $\sum_{s\in A\setminus \{q\}} \frac{k_s}{k_{q}^-} = 1$,
and therefore $ALG \leq 2$. We denote $\ell$ the level of the edge containing
$q$. To estimate the change of the potential, we again consider two cases.
\begin{enumerate}
\item {\em When  $k_q^- \geq \beta_\ell h_q^-$.} Here the movement of $q$ increases
the excess in the edge containing $q$. Let us denote $E$ (resp.~$D$) the servers of
$A\setminus\{q\}$ residing in the excessive resp. deficient edges.
By taking the largest possible $\alpha^D$ and the  smallest possible $\alpha^E$
coefficient in the estimation, we can bound the change of the potential due to
the move of the servers in $T_q^-$ as,
\[ \alpha^D_{\ell-1} \sum_{s\in D} \frac{k_s}{k_q^-}
	- \alpha^E_{\ell-1} \sum_{s\in E} \frac{k_s}{k_q^-}
	\leq \frac1\beta \alpha^D_{\ell-1} - \frac1\gamma \alpha^E_{\ell-1},
\]
where the above inequality holds due to the Lemma \ref{lem:ex_ph2}.

As the movement of $q$ itself causes an increase of $\Phi$ by
$\alpha^E_\ell$, we have
\[ ALG + \Delta\Phi \leq 2
	+ \alpha^D_{\ell-1}/\beta
	- \frac1\gamma \alpha^E_{\ell-1}
	+ \alpha^E_{\ell} .
\]
To see that this is non-positive,  recall that
$\alpha^E_\ell = \AlphaEL$. Therefore
\begin{align*}
\frac1\gamma \alpha^E_{\ell-1}
&= \frac1\gamma \sum_{i=\ell-1}^{d-1}
	\gamma^{i-\ell+2}
	\big(2 + \frac1\beta \alpha^D_i\big)
	+ \frac1\gamma
	\gamma^{d-\ell+1} \alpha^E_d\\
&= \big(2 + \frac1\beta \alpha^D_{\ell-1}\big)
	+ \sum_{i=\ell}^{d-1} \gamma^{i-\ell+1}
		\big(2 + \frac1\beta \alpha^D_i\big)
	+ \gamma^{d-\ell} \alpha^E_d\\
&= 2 + \frac1\beta \alpha^D_{\ell-1} + \alpha^E_\ell.
\end{align*}
\item
 {\em When $k_q^- < \floor{\beta_\ell h_q^-}$.}
This case is much simpler.
All the movement inside of $T_q^-$ might contribute to the increase of
deficiency at level at most $\ell-1$.
On the other hand, $q$ then causes a decrease of deficiency at level
$\ell$ and we have $ ALG + \Delta\Phi \leq 2 + \alpha^D_{\ell-1} - \alpha^D_\ell.$
This is less or equal to $0$, as
$\alpha^D_\ell \geq \alpha^D_{\ell-1} + 2$.
\end{enumerate}
\end{proof}


\begin{lemma}\label{lem:alg_ratio}
For each $1\leq \ell\leq d$, both $\ceil{\beta_\ell} \alpha^D_\ell$ and
$\ceil{\beta_\ell} \alpha^E_\ell$ are of order $O(d\cdot \gamma^{d+1})$.
\end{lemma}

\begin{proof}[Proof of Lemma \ref{lem:alg_ratio}]
We have defined $\alpha^D_\ell = 2\ell-1$, and therefore
\[ \ceil{\beta_\ell} \alpha^D_\ell \leq 2\cdot \beta^\ell (2\ell-1)
	\leq \beta^d (2d-1) = O(d\cdot \gamma^{d+1}),
\]
since we have chosen $1<\beta\leq 2$ and $\gamma \geq 2$ for any possible
$\delta$.

Now we estimate the value of $\alpha^E_d$.
Since $\delta \geq \beta \cdot \beta_d$ and
$\frac{\delta}{\delta-\beta_d} \leq \frac1{1-\beta_d/\delta}
\leq \frac1{1-1/\beta} = \gamma$, we have
\[ \alpha^E_d = \AlphaEd \leq \gamma (3 + \frac1\beta \alpha^D_d),
\]
which is $O(d \cdot \gamma)$, as $\beta > 1$ and $\alpha^D_d = O(d)$.
For $\alpha^E_\ell$, we have
\[ \alpha^E_\ell \leq
	(2 + \frac1\beta \alpha^D_d) \sum_{i=\ell}^{d-1} \gamma^{i-\ell+1}
	+ \gamma^{d-\ell} \alpha^E_d 
	\leq \gamma^{d-\ell+1} \big(2 + \frac{2d-3}\beta + \alpha^E_d \big)
	= O(\gamma^{d-\ell+2} \cdot d),
\]
because $1<\beta \leq 2$ and $\alpha^E_d = O(d)$. The second inequality follows
by summing-up the geometric series and using the fact that $\gamma
\geq 2$.
Therefore, as $\beta_\ell \leq
\gamma^{\ell-1}$, we have $\ceil{\beta_\ell} \alpha^E_\ell = O(\gamma^{d+1} d)$,
and this concludes the proof.
\end{proof}

%% file: conclusion.tex
Several intriguing open questions remain, and we list some of them here.

\begin{itemize}
\item
Is the dependence on $d$ in Theorem~\ref{thm:alg-ub} necessary?
While Theorem \ref{thm:gen_lb} gives a separation between depth-$1$
and depth-$2$ HSTs, it is unclear to us whether a lower bound which increases substantially with depth is possible.
Note that a lower bound of $g(d)$ for depth $d$,
where $g(d) \rightarrow \infty$ as
$d \rightarrow \infty$ (provided that $h$ is large enough),
would be very surprising.
This would imply that there is
no $O(1)$-competitive algorithm for general metric spaces. 
\item
Can we get an $o(h)$-competitive algorithm for other metric spaces?
An interesting metric is the line:
Both DC and WFA have competitive ratio $\Omega(h)$ in the line,
and we do not even know any good candidate algorithm.
Designing an algorithm with such a guarantee would be very interesting.
Also, the only lower bound known for the line is 2 for $h=2$.
It would be interesting to get a non-trivial lower bound for arbitrary $h$. 
\end{itemize}

%% file: appendix.tex
\section{Algorithm}

\subsection{Proofs of lemmas from Section~\ref{sec:algo}}\label{sec:alg_omit}

\begin{proof}[Proof of Lemma \ref{lem:ex_ph1}]
During the Phase 1, each server of the algorithm resides either in $T_u$ or
in $T_s$ for some $s\in E \cup D$; therefore
$k = k_u + \sum_{s\in E} k_s + \sum_{s\in D} k_s$.
However, for each $s\in D$, we have
$k_s \leq \floor{\beta_{\ell(s)} \cdot h_s}$,
otherwise $E_e$ would be positive for the edge $e$ containing $s$.
Therefore we have
\[ \sum_{s\in D} k_s
	\leq \sum_{s\in D} \floor{\beta_{\ell(s)} h_s}
	\leq \sum_{s\in D} \floor{\beta_d h_s}
	\leq \sum_{s\in D} \beta_d \cdot h_s \leq \beta_d (h-h_u).
\]

To prove the second inequality, observe that
\[ \sum_{s\in E} k_s = k - k_u - \sum_{s\in D} k_s
	\geq k - k_u - \beta_d (h-h_u).
\]
However, we assumed that $k_u \leq \floor{\beta_2 h_u} \leq \beta_d h_u$,
and therefore $\sum_{s\in E} k_s \geq k - \beta_d \cdot h$.
\end{proof}

\medskip

\begin{proof}[Proof of Lemma \ref{lem:ex_ph2}]
In this lemma, we crucially use the fact that
$\beta_\ell = \beta \cdot \beta_{\ell-1}$.
Similarly to the previous proof, we have
$k_s \leq \lfloor\beta_{\ell(s)}\cdot h_s\rfloor \leq \beta_{\ell-1}\cdot h_s$
for each server $s\in D$, since $s$ can be located in a level at most $\ell-1$.
However, in this case we claim that
\[ \sum_{s\in D} k_s \leq \sum_{s\in D} \beta_{\ell-1} \cdot h_s
	\leq \beta_{\ell-1} (h_q^- - 1).
\]
This is because we assume that adversary has already served the request
and some server $a$, one of his $h_q^-$ servers in $T_q^-$,
is already at the requested point.
Since no online servers reside in the path between $q$ and the
requested point, $a$ does not belong to $T_s$ for any $s\in D\cup E$.
Therefore we have
$\sum_{s\in D} h_s \leq \sum_{s\in (D\cup E)}h_s \leq h_q^- - 1$.
To finish the proof of the first inequality, observe that our assumption implies
\[ k_q^- \geq \floor{\beta \beta_{\ell-1} h_q^-}
	\geq \beta \beta_{\ell-1} h_q^- - 1
	\geq \beta (\beta_{\ell-1} h_q^- - 1).
\]
Therefore we have
$ \beta_{\ell-1} (h_q^- - 1) \leq \beta_{\ell-1} h_q^- - 1 \leq k_q^- / \beta$.

For the second inequality, note that $\frac1\gamma = (1 - \frac1\beta)$.
Since $k_q^- = \sum_{s\in D} k_s + \sum_{s\in E} k_s$, we have
\[ \sum_{s\in E} k_s \geq k_q^- - \frac1\beta k_q^- = \frac1\gamma k_q^-. \]
\end{proof}

\subsection{Algorithm for bounded-diameter trees}\label{sec:alg_diam}

Since Algorithm~\ref{alg:ub} works for depth-$d$ trees with arbitrary
edge lengths, we can embed any diameter-$d$ tree into a depth-$d$ tree
with a very small distortion
by adding fake paths of short edges to all nodes.

More precisely, 
let $T$ be a tree of diameter $d$ with arbitrary edge lengths, and
let $\alpha$ be the length of the shortest edge of $T$ (for any finite $T$ such
$\alpha$ exists). We fix $\epsilon > 0$ a small constant.
We create an embedding $T'$ of $T$ as follows.
We choose the root $r$ arbitrarily, and to
each node $v\in T$ such that the path from $r$ to $v$ contains $\ell$ edges, we
attach a path containing $d-\ell$ edges of total length $\epsilon \alpha/2$.
The leaf at the end of this path we denote $v'$. We run Algorithm~\ref{alg:ub}
in $T'$ and each request at $v\in T$ we translate to $v' \in T'$.
We maintain the correspondence between the servers in $T$ and the servers in
$T'$, and the same server which is finally moved to $v'$ by
Algorithm~\ref{alg:ub}, we also use to serve the request $v\in T$.

For the optimal solutions on $T$ and $T'$ we have
$(1+\epsilon)OPT(T) \geq OPT(T')$, since  any feasible solution in $T$ we can
be converted to a solution in $T'$ with cost at most $(1+\epsilon)$ times
higher.
By Theorem~\ref{thm:alg-ub2}, we know that the cost of Algorithm~\ref{alg:ub} in
$T'$ is at most $R \cdot OPT(T')$, for $R=\Theta(d\cdot \gamma^{d+1})$, and
therefore we have $ALG(T') \leq (1+\epsilon) R \cdot OPT(T)$.

%% file: lbs1.tex
In this section we prove Theorems \ref{thm:gen_lb} and \ref{thm:lb-wfa}. We first show a general lower bound on the competitive ratio of any algorithm for depth-2 HSTs. Then we give lower bound on the competitive ratio of WFA.

\subsection{General lower bound for depth-2 HSTs}\label{sec:gen_lb}
We now give a lower bound on the competitive ratio of any deterministic online algorithm on depth-2 HSTs. In particular, we show that for sufficiently large $h$, any deterministic online algorithm has competitive ratio at least $2.4$. 

The metric space is a depth-2 HST $T$ with the following properties: $T$ contains at least $k+1$ elementary subtrees and each one of them has
at least $h$ leaves. To ease our calculations, we assume that edges of the lower level have length $\epsilon \ll 1  $ and edges of the upper level have length $1-\epsilon$. So the distance between leaves of different elementary
subtrees is $2$.

\newrestatedthm{thm:gen_lb}
\begin{theorem-thm:gen_lb}[restated] For sufficiently large $h$, even when $k/h \rightarrow \infty$, there is no 2.4-competitive
deterministic algorithm for general metrics, and in particular even for depth-$2$ HSTs.
\end{theorem-thm:gen_lb}

\begin{proof} 
For a level $1$ node $u$, let $T_u$ denote the elementary subtree rooted at $u$.
We assume without loss of generality that 
all the offline and online servers are always located at the leaves. 
We say that a server is inside subtree $T_u$,
if it is located at some leaf of $T_u$.
If there are no online servers at the leaves of $T_u$, we say that $T_u$ is \textit{empty}.
Observe that at any given time there exists at least one empty elementary subtree.

Let $\mathcal{A}$ be an online algorithm. The adversarial strategy consists of arbitrarily many iterations of a phase. During a phase,
some offline servers are moved to an empty elementary subtree $T_u$ and requests
are made there until the cost incurred by $\mathcal{A}$ is sufficiently large. At this point the phase ends and a new phase may start in another empty subtree. Let ALG and ADV denote the cost of $\mathcal{A}$ and adversary respectively during a phase. We will ensure that for all phases ALG $\geq 2.4 \cdot \mbox{ADV}$. This implies the lower bound on the competitive ratio of $\mathcal{A}$.  

We describe a phase of the adversarial strategy. The adversary moves some $\ell \leq h$ servers to the empty elementary subtree $T_u$ and makes requests at 
leaves of $T_u$ until $\mathcal{A}$ brings $m$ servers there. In particular, each request appears at a leaf of $T_u$ that is not occupied by a server of $\mathcal{A}$. We denote by $s(i)$ the cost that $\mathcal{A}$ has to incur for serving requests inside $T_u$ until it moves its $i$th server there (this does not include the cost of moving the server from outside $T_u$). Clearly, $s(1)=0$ and $s(i) \leq s(i+1)$ for all $i>0$. We restrict our attention to $i \leq h$. The choice of $\ell$ and $m$ depends on the values $s(i)$ for $2 \leq i \leq h$. We will now show that for any values of $s(i)$'s, the adversary can choose $\ell$ and $m$ such that ALG $\geq 2.4 \cdot \mbox{ADV}$.

 First, if there exists an $i$ such that $s(i) \geq 3i$, we set $\ell = m = i$. 
Intuitively, the algorithm is too slow in bringing his servers in this case.
Both $\mathcal{A}$ and the adversary incur a cost of $2i$ to move $i$ servers to $T_u$. However, $\mathcal{A}$ pays a cost of $s(i)$ for serving requests inside $T_u$, while the adversary can serve all those requests at zero cost (all requests can be located at leaves occupied by offline servers). Overall, the cost of $\mathcal{A}$ is $2i + s(i) \geq 5 i$, while the offline cost is $2i$. Thus we get that $\mbox{ALG}\geq 2.5 \cdot \mbox{ADV}$. 

Similarly, if $s(i) \leq (10 i - 24)/7$ for some $i$, we choose
$\ell = 1$ and $m=i$. Roughly speaking, in that case the algorithm is too
``aggressive'' in bringing its first $i$ servers, thus incurring a large
movement cost. Here, the adversary only moves one server to $T_u$.
Each request is issued at an empty leaf of $T_u$.
Therefore $\mathcal{A}$ pays for each request in $T_u$ and the same holds for
the single server of the adversary.

So, ALG $= 2i + s(i)$ and ADV $= 2 + s(i)$. By our assumption on $s(i)$, this
gives 
\[\mbox{ALG} - 2.4 \cdot \mbox{ADV} = 2i + s(i) - 4.8 - 2.4 \cdot s(i)
	= 2i - 4.8  - 1.4\cdot s(i)  \geq 0 \]

We can thus restrict our attention to the case that
$s(i) \in (\frac{10 i - 24}{7},3i)$, for all $ 2 \leq i\leq h$. Now, for $1 < \ell < h$ and $m=h$, we can upper bound the offline cost as follows:

\begin{equation*}\label{lb-ratio1}
\mbox{ADV } \leq 2 \ell + (s(h) - s(\ell)) \cdot \frac{h-\ell+1}{h} 
\end{equation*}      

The first term is the cost of moving $\ell$ servers to $T_u$. The second term upper bounds the cost for serving requests in $T_u$ for the following reason: the adversary has to pay only during the part when $\mathcal{A}$ has at least $\ell$ servers at $T_u$. The cost of the adversary during this time equals the cost of $\mathcal{A}$ (which is $s(h) - s(\ell)$) divided by its competitive ratio on a uniform metric (which is clearly larger than $\frac{h}{h-\ell+1}$). 
Let us denote $c = s(h)/h$. Note that assuming $h$ is large enough, $c \in (1,3)$. We get that

\begin{equation}\label{lb-ratio2}
\frac{\mbox{ALG}}{\mbox{ADV}}
\geq \frac{2h + s(h)}{2 \ell + (s(h) - s(\ell)) \cdot \frac{h-\ell+1}{h}}
\geq \frac{2h + ch}{2 \ell +(ch - \frac{10\ell-24}{7})\cdot\frac{h-\ell+1}{h}}
\end{equation}
We now show that for every value of $c \in (1,4)$, there is an $\ell = \beta h$,
where $\beta$ is a constant depending on $c$, such that the right hand side of \eqref{lb-ratio2} is
at least $2.4$.
First, as $h$ is large enough, the right hand side is arbitrarily close to 
\[ \frac{2h + ch}{2 \ell +(ch - 10\ell/7)\cdot (1-\ell/h)}
	= \frac{2+c}{2 \beta + (c-10\beta/7)(1-\beta)}
	= \frac{2+c}{10/7 \beta^2 + (4/7-c) \beta + c}
\]
We choose $\beta := (c-4/7)/(20/7)$.
Note that, as $c \in (1,3)$, we have that $\beta < 1$ and hence $\ell < h$.
The expression above then evaluates to
$(2+c)/(c - (c-\frac47)^2/\frac{40}7)$.
By standard calculus, this expression is minimized at
$c= (2\sqrt{221}-14)/7$  where it attains a value higher than $2.419$.

%
%

\end{proof}


%% file: lbs3.tex
\subsection{Lower bound for WFA on depth-3 HSTs}\label{sec:lb-wfa-hst}

We give an $\Omega(h)$ lower bound on the competitive ratio of the Work Function Algorithm (WFA) in the  $(h,k)$-setting. More precisely, we show a lower bound of $h+1/3$ for $k=2h$.
We first present the lower bound for a depth-$3$ HST. Later we show how the construction can be adapted to work for the line. We also show that this exactly matches the upper bound of $(h+1)\cdot \mbox{OPT}_h - \mbox{OPT}_k$ implied by results of Koutsoupias \cite{Kou99} for the line.

The basic idea behind the lower bound is to trick the WFA to use only $h$ servers for servicing requests in an ``active region'' for a long time before it moves its extra available online servers. Moreover, we make sure that during the time the WFA uses $h$ servers in that region, it incurs an $\Omega(h)$ times higher cost than the adversary. Finally, when WFA realizes that it should bring its extra servers, the adversary moves all its servers to some different region and starts making requests there. 
So eventually, WFA is unable to use its additional servers in a useful way to improve its performance.

\medskip\noindent 
{\bf Theorem~\ref{thm:lb-wfa} } {\em (restated) The competitive ratio of WFA in a depth-$3$ HST for $k=2h$ is at least
$h+1/3$.
}

\begin{proof}

Let $T$ be a depth-3 HST. We assume that the lengths of the edges in a root-to-leaf path are $\frac{1-\epsilon}{2}, \frac{\epsilon - \epsilon'}{2}, \frac{\epsilon'}{2}$, for $\epsilon' \ll \epsilon \ll 1$. So the diameter of $T$ is 1 and the diameter of depth-2 subtrees is $\epsilon$. Without loss of generality we will also assume that $\frac{2}{\epsilon}$ is integer. Let $L$ and $R$ be two subtrees of depth 2. Inside each one of them, we focus on 2 elementary subtrees $L_1,L_2$ and $R_1,R_2$ respectively (see figure~\ref{fig:lb-wfa}). All requests appear at leaves of those four elementary subtrees. Note that both the WFA and the adversary always have their servers at leaves. This way, we say that some servers of the WFA (or the adversary) are inside a subtree, meaning that they are located at some leaves of that subtree.  

\begin{figure}[t!]
\hfill\includegraphics{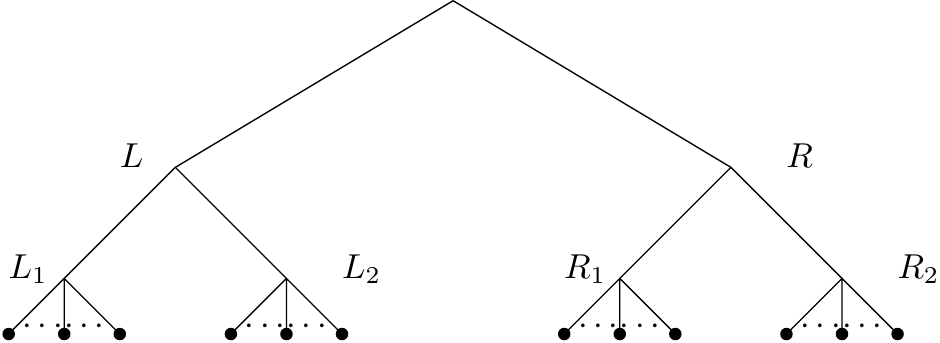}\hfill\ %
\caption{The tree where the lower bound for WFA is applied: All requests arrive at leaves of $L_1,L_2,R_1$ and $R_2$.}
\label{fig:lb-wfa}
\end{figure}

  The adversarial strategy consists of arbitrary many iterations of a phase. At the beginning of the phase, the WFA and the adversary have all their servers in the same subtree, either $L$ or $R$. At the end of the phase, all servers are moved to the other subtree ($R$ or $L$ resp.), so a new phase may start. Before describing the phase, we give the basic strategy, which is repeated many times during a phase. For any elementary subtree $T$, we define strategy $S(T)$.

\vskip2mm\noindent
\textit{Strategy $S(T)$:}
\begin{enumerate}
\item Adversary moves all its $h$ servers to leaves $t_1,t_2,\dotsc,t_h$ of $T$.
\item While the WFA has $i < h $ servers in $T$: Request points $t_1,t_2,\dotsc,t_{i+1}$ in an adversarial way (i.e each time request a point where the WFA does not have a server) until $(i+1)$th server arrives at $T$.
\end{enumerate}


We now describe a left-to-right phase, i.e., a phase which starts with all servers at $L$ and ends with all servers at $R$. A right-to-left phase is completely symmetric, i.e., we replace $R$ by $L$ and $R_i$ by $L_i$. 

\vskip2mm\noindent
\textit{Left-to-right phase:}
\begin{itemize}
\item[--] \textit{Step 1:} Apply strategy $S(R_1)$ until WFA moves $h$ servers to $R_1$.
\item[--] \textit{Step 2:} Repeat $S(R_2)$ and $S(R_1)$ until WFA moves all its $k$ servers to $R$.
\end{itemize}








We now state two basic lemmas for counting the online and offline cost during each step. Proofs of those lemmas, require a characterization of work function values, which comes later on. Here we just give the high-level idea behind the proofs. Let ALG and ADV denote the online and offline cost respectively. 

\begin{lemma}
\label{lem:step1-cost}
During Step 1, ALG$=h^2$ and ADV$=h$.
\end{lemma}

Intuitively, we will exploit the fact that the WFA moves its servers too slowly towards $R$ as long as $\ell < h$. In particular, we show that whenever the WFA has $\ell$ servers in $R$, it incurs a cost of $2\ell$ inside $R$ before it moves its $(\ell+1)$th server there. This way we get that the cost of the algorithm is $\sum_{\ell=1}^{h-1} 2 \ell + h = h^2$. The fact that adversary pays $h$ is trivial; the adversary just moves its $h$ servers by distance 1 and then serves all requests at zero cost. 

\begin{lemma}
\label{lem:step2-cost}
During Step 2, ALG $\geq (2-2\epsilon)h^2 +h $ and ADV $\leq(2+\epsilon) h$.
\end{lemma}

 Roughly speaking, to prove this lemma we make sure that for almost the whole Step 2, the WFA has $h$ servers in $R$ and they incur a cost $h$ times higher than the adversary. The additional servers move to $R$ almost at the end of the phase.

The theorem follows from lemmata~\ref{lem:step1-cost} and~\ref{lem:step2-cost}. Summing up for the whole phase, the offline cost is $(3+\epsilon)\cdot h$ and the online cost is $3(3-2\epsilon)h^2+h$. We get that 

$$ \frac{ALG}{ADV} \geq \frac{h^2 + (2-2\epsilon)h^2 +h }{h + (2+\epsilon) h} =  \frac{(3-2\epsilon)h^2+h}{(3+\epsilon) h} \rightarrow h + \frac{1}{3}   $$
\end{proof}

\subsubsection*{Work Function values}

We now give a detailed characterization of how the work function evolves during left-to-right phase. For right-to-left phases the structure is completely symmetric.  

\smallskip

\noindent {\bf Basic Properties:} Recall that for any two configurations $X,Y$ at distance $d(X,Y)$ the work function $w$ satisfies, $$  w(X) \leq w(Y) + d(X,Y)  .$$ 
Also, let $w$ and $w'$ the work function before and after a request $r$ respectively. For a configuration $X$ such that $r \in X$, we have that $w'(X) = w(X)$. In the rest of this section, we use these properties of work functions without further explanation.  

\smallskip

\noindent {\bf Notation:}  Let $(L^{i},R^{k-i})$ denote a configuration which has $i$ servers at $L$ and $(k-i)$ servers at $R$. Let $w(L^i,R^{k-i})$ the minimum work function value of a configuration $(L^i,R^{k-i})$. If we need to make precise how many servers are in each elementary subtree, we denote by $(L^{i},r_1,r_2)$ a configuration which has $i$ servers at $L$, $r_1$ servers at $R_1$ and $r_2$ servers at $R_2$. Same as before,  $w(L^{i},r_1,r_2)$ denotes the minimum work function value of those configurations.


\medskip

\noindent {\bf Initialization:} For convenience we assume that at the beginning of the phase the minimum work function value of a configuration is zero. Note that this is without loss of generality: If $m = \min_{X} w(X)$ when the phase starts, we just subtract $m$ from the work function values of all configurations.  We require the invariant that at the beginning of the phase, for $0 \leq i \leq k$ :  

\begin{equation}
\label{eq:start_phase}
w(L^{k-i},R^{i}) =  i  .
\end{equation}

This is clearly true for the beginning of the first phase, as all online servers are initially at $L$. We are going to make sure, that the symmetric condition holds at the end of the phase, so a right-to-left phase may start.

\noindent {\bf First Step:} We now focus on the first step of the phase, i.e.~the first execution of $S(R_1)$. The next lemma characterizes the structure of the work function each time the WFA decides to move a new server to $R$.


\begin{lemma}
\label{lem:wf_l2}
At the moment when the WFA moves its $\ell$th server from $L$ to $R$, for any $1 \leq \ell < h$, the following two things hold: 

\begin{enumerate}[(i)]
\item For all $0\leq i <\ell$, $w(L^{k-i},R^{i}) = w(L^{k-\ell},R^{\ell}) + (\ell - i)$
\item For all $\ell \leq i \leq k$, $w(L^{k-i},R^{i}) = w(L^{k-\ell},R^{\ell}) + (i - \ell)$ 
\end{enumerate}
In other words, having $\ell$ servers in $R$ is the lowest state, and all other states are tight with respect to it.
\end{lemma}

\begin{proof}
We use induction on $\ell$.

\textit{Induction Base:} In the base case when $\ell=0$, the first part of the lemma is vacuously true. The second part of the Lemma holds by the assumed invariant \eqref{eq:start_phase}.

\textit{Induction Step:} Assume that this is true for some $ 0 \leq \ell < h-1$. We are going to show that the lemma holds for $\ell +1$. Let $w$ be the work function at the time when the $\ell$th server arrives at $R$ and $w'$ be the work function at the time when the $(\ell+1)$th server arrives at $R$.

\begin{enumerate}[(i)]

\item We first show that the lemma holds for $i=\ell$: By construction of the WFA, at the moment it moves its $(\ell+1)$th server at $R$, the following equation holds: 
\begin{equation}
\label{eq:wf_ell+1}
 w'(L^{k-\ell},R^{\ell}) = w'(L^{k-(\ell+1)},R^{\ell+1}) +1,
\end{equation}
 
We now focus on $i < \ell$. By induction hypothesis, we have that $w(L^{k-i},R^{i}) = w(L^{k-\ell},R^{\ell}) + (\ell - i)$. During the time when WFA has $\ell$ servers at $R$, all requests are located in $R$. Thus, $w(L^{k-i},R^{i})$  increases at least as much as $w(L^{k-\ell},R^{\ell})$, i.e.~
\[w'(L^{k-i},R^{i}) - w(L^{k-i},R^{i}) \geq w'(L^{k-\ell},R^{\ell}) - w(L^{k-\ell},R^{\ell}). \]
So, we have that 
\begin{eqnarray*}
w'(L^{k-i},R^{i}) & \geq &  w'(L^{k-\ell},R^{\ell}) - w(L^{k-\ell},R^{\ell}) +  w(L^{k-i},R^{i}) \\
                  &  = &  w'(L^{k-\ell},R^{\ell})  + (\ell - i).
\end{eqnarray*}                  
Clearly, as $ w'(L^{k-i},R^{i}) \leq w'(L^{k-\ell},R^{\ell}) + (\ell - i)$. This implies that equality holds, and we get that
\[w'(L^{k-i},R^{i}) = w'(L^{k-\ell},R^{\ell}) + (\ell - i) =  w(L^{k-\ell-1},R^{\ell+1}) + (\ell + 1 -i),\]
    where the last equality holds due to (\ref{eq:wf_ell+1}).
\item For $i \geq (\ell+1)$: Since the beginning of the phase there are $(\ell+1)$ points requested, so $w(L^{k-i},R^i)$ does not increase. By \eqref{eq:start_phase}, we get that 
\[w'(L^{k-i},R^{i}) = i = \ell + 1 +(i - (\ell+1)) = w'(L^{k-\ell-1},R^{\ell+1}) + i - (\ell +  1) .\]
\end{enumerate}
\end{proof}









\noindent {\bf Second step:} By lemma~\ref{lem:wf_l2}, at the beginning of the second step, the work function values satisfy: 
\begin{equation}
\label{eq:step2_in1}
 w(L^{2h-i},R^{i}) = w(L^h,R^h) + (h - i),  \mbox{  for } 0 \leq i \leq h,
\end{equation}
\begin{equation}
\label{eq:step2_in2}
 w(L^{2h-i},R^{i}) = w(L^h,R^h)+ (i - h), \mbox{  for  } h \leq i \leq 2h,
\end{equation}

We first claim that (\ref{eq:step2_in1}) holds for the entire second step. This follows as all the requests arrive at $R$, so the optimal way to serve those requests using $i<h$ servers, can only have larger cost than using $h$ servers.


  We now consider executions of $S$ during the time that WFA has $h$ servers in $R$. We describe the changes in work function values of configurations $(L^h,r_1,r_2)$ during an execution of $S(R_2)$ and for $S(R_1)$ the situation will be completely symmetric. We require that initially for all $i \leq h$,  
\begin{equation}
\label{eq:step_2_h}
 w(L^h,h-i,i) = w(L^h,h,0) + \epsilon \cdot i.
\end{equation}  

This is true at the beginning of the second step, and we are going to make sure that at the end of $S(R_2)$ the symmetric is true, so an execution of $S(R_1)$ can start. Similarly to lemma~\ref{lem:wf_l2} we can show the following:


\begin{lemma}
\label{lem:wf_l3}
Consider an execution of $S(R_2)$ such that WFA has exactly $h$ servers in $R$. At the moment when the $\ell$th server moves to $R_2$, the following two things hold:

\begin{enumerate}[(i)]
\item For all $0\leq i <\ell$, $w(L^h,h-i,i) = w(L^h, h-\ell,\ell) + \epsilon (\ell - i)$

\item For all $\ell \leq i \leq h$, $w(L^h,h-i,i) = w(L^h,h-\ell,\ell) + \epsilon (i - \ell)$ 

\end{enumerate}

\end{lemma}

\begin{proof}
The proof is exactly the same as the proof of lemma~\ref{lem:wf_l2}, we just need to replace  $(L^{k-i},R^i)$ by $(L^h,h-i,i)$ and  scale all distances by $\epsilon$.
\end{proof}

We get that when $h$th server arrives to $R_2$, $ w(L^h,i,h-i) = w(L^h,0,h) + \epsilon \cdot i$ for all $ i \leq h$. Observe that this is the symmetric version of~\eqref{eq:step_2_h}. That means, the work function satisfies the initial requirement for $S(R_1)$ to start.








Now, we come to harder part where we want to argue that the WFA does not bring in extra servers. To this end, we will investigate the changes in $w(L^{h-i},R^{h+i})$ during executions of second step.

\begin{lemma}
\label{lem:Deltaw}
Consider an execution of $S$ such that for all its duration WFA has at most $h+i$ servers in $R$, for $0\leq i \leq (h-1)$. During such an execution, $\Delta w(L^{h-i},R^{h+i}) = \epsilon (h - i)$.
\end{lemma}

\begin{proof} 

We prove the lemma for executions of $S(R_2)$. We focus on configurations $(L^{h-i}, R^{h+i})$. Initially, among those configurations, a configuration $(L^{h-i},h, i)$ has the minimum work function value. This is clearly true for the first execution of second step and we will show that it holds for all subsequent executions. Moreover for all $j\leq h$, the following equation holds: 
\begin{equation}
\label{eq:consist}
 w(L^{h-i},h-j,i+j) = w(L^{h-i},h,i) + j \cdot \epsilon
\end{equation}   
  Up to the time that at most $i$ points of $R_2$ are requested, $w(L^{h-i},h,i)$ does not increase. After the time that $i$th WFA server arrives at $R_2$, we can apply lemma~\ref{lem:wf_l3} for $\ell = i $ to $h$. In particular, $w(L^{h-i}, h-\ell + i , \ell)$ increases in the same way as $w(L^h,h-\ell,\ell)$. 

At the end, a configuration $(L^{h-i}, i, h )$ has the minimum work function value. Moreover, we have that $w'(L^{h-i}, i, h ) = w(L^{h-i}, i, h )$, because exactly $h$ points of $R_2$ are requested. For $w(L^{h-i}, i, h)$ we can apply equation~\eqref{eq:consist} for $j= h-i$. We get that 

  \[w'(L^{h-i},R^{h+i}) = w'(L^{h-i}, i, h ) = w(L^{h-i}, i, h ) = w(L^{h-i},R^{h+i}) + \epsilon \cdot (h-i)\]
     

\end{proof}

Next two lemmata show that we can force the algorithm use only $h$ servers in $R$ for a long time and that additional servers arrive roughly at the same time.

 Let $N = 2/ \epsilon$. 

\begin{lemma}
\label{lem:step2_count_runs}
After $N - 2$ executions of $S$ in Step 2, the WFA still has $h$ servers in $R$
\end{lemma}

\begin{proof}
 Using lemma~\ref{lem:Deltaw} we get that after $(N-2)$ executions of $S$ in Step 2,
 \begin{equation}
\label{eq:Nth0}
w(L^h,R^h)  = h + (N-2) \cdot \epsilon \cdot h = 3h - 2\epsilon \cdot h, 
\end{equation}
     
\begin{equation}
\label{eq:Nth1}
\mbox{ and  }w(L^{h-1},R^{h+1}) = h + 1 + (N-2) \cdot \epsilon (h-1) = 3h - 1- 2\epsilon (h-1). 
\end{equation}

By~\eqref{eq:Nth0} and ~\eqref{eq:Nth1} we get that $ w(L^{h-1},R^{h+1}) + 1 > w(L^h,R^h) + \epsilon h$. That means, that up to the end of $(N-2)$th execution of $S$, the WFA would always prefer to move a server inside $R$ by distance $\epsilon$ rather than moving a server from $L$ by distance 1. This clearly shows that the WFA is still in some configuration $(L^h,R^h)$.




\end{proof}

\begin{lemma}
Step 2 consists of at most $N +1$ executions of $S$.
\end{lemma}

\begin{proof}
Assume that after $N$ executions of $S$ in Step 2, the WFA has $h+i$ servers in $R$, for some $0\leq i \leq (h-1)$. In other words, it is in some configuration $(L^{h-i},R^{h+i})$. We can apply lemma~\ref{lem:Deltaw} for $(L^{h-i},R^{h+i})$ for all $N$ executions. After $N$ executions of $S$ in Step 2 

\begin{equation}
\label{eq:end_second1}
\Delta w(L^{h-i}, R^{h+i}) =  2/\epsilon \cdot \epsilon \cdot (h-i) = 2(h-i)
\end{equation}

 By adding \eqref{eq:step2_in2} and \eqref{eq:end_second1} we get that
\begin{equation}
\label{eq:end_second2} 
   w'(L^{h-i}, R^{h+i}) = w(L^{h-i}, R^{h+i}) + 2(h - i) = h + i +2h - 2i= 2h + (h- i).
\end{equation}

Since the beginning of the phase, exactly $2h$ points in $R$ are requested. Thus, $w'(L^0,R^{2h}) = w(L^0,R^{2h}) = 2h$; during the whole phase. This way~\eqref{eq:end_second2} is equivalent to  $ w'(L^{h-i}, R^{h+i}) = w'(L^0,R^{2h})+(h-i) $. That implies that during the next execution of $S$, the WFA moves to the configuration $(L^0,R^{2h})$, i.e brings all its servers at $R$, so the phase ends.

\end{proof}

\noindent {\bf End of the phase:} By \eqref{eq:step2_in1} and \eqref{eq:end_second2} we get that in the end of the phase $w(L^{i},R^{2h-i}) = 2h + i $ for all $0\leq i \leq 2h$. So a new right-to-left phase may start, satisfying the initial assumption about the structure of work function values.

\subsubsection*{Bounding Costs}

Now, we are ready to bound online and offline costs during the phase. 

\begin{lemma}
\label{lem:ins_cost}
During the time when the WFA uses $\ell$ servers in $R$, it incurs a cost of $2\ell$.
\end{lemma}

\begin{proof}
 Let $w$ the work function at time when $\ell$th server arrives at $R$ and $w'$ the work function at time when $(\ell+1)$th arrives there. From lemma~\ref{lem:wf_l2} we get that $$ w'(L^{k-\ell},R^{\ell}) - w(L^{k-\ell},R^{\ell}) = w'(L^{k-\ell-1},R^{\ell+1}) +1 - w(L^{k-\ell-1},R^{\ell+1}) =   2 . $$
That means, the optimal way to serve all requests during this time using $\ell$ servers costs 2 (this holds because this configuration is in the support of work function). All those requests are placed into an elementary subtree, which is equivalent to the paging problem. By~\cite{ST85} we get that for that request sequence, any online algorithm using $\ell$ servers incurs a cost at least $\ell$. Thus, WFA pays at least $2\ell$ inside $R$ during the time it has $\ell $ servers there. 
\end{proof}

\medskip\noindent
{\bf Lemma~\ref{lem:step1-cost} } {\em (restated) During Step 1, ALG$=h^2$ and ADV$=h$. 
}

\begin{proof}
Clearly, ADV$=h$; the adversary moves $h$ servers by distance 1 and then serves all requests at zero cost. By lemma~\ref{lem:ins_cost} we get that WFA incurs a cost $\sum_{\ell=1}^{h-1} 2 \ell$ inside $R$. Moreover, it pays a movement cost of $h$ to move its $h$ servers from $L$ to $R$. Overall, we get that the online cost during Step 1 is $\sum_{\ell=1}^{h-1} 2 \ell + h = h^2$.


\end{proof}

\begin{lemma}
\label{lem:cost-step2-h}
Consider an execution of $S$ during Step 2, where the WFA uses only the $h$ servers in $R$ to serve the requests. For such an execution, ALG$=\epsilon h^2 $ and ADV$=\epsilon h$.
\end{lemma}

\begin{proof}
Same as proof of lemma~\ref{lem:step1-cost} , just scale everything by $\epsilon$.
\end{proof}

\medskip\noindent
{\bf Lemma~\ref{lem:step2-cost} } {\em (restated)  During Step 2, ALG$\geq (2-2\epsilon)h^2 +h $ and ADV$\leq(2+\epsilon) h$.
}

\begin{proof}
Second step consists of at most $2 / \epsilon +1$ executions of $S$, where in each one of them the adversary incurs a cost $\epsilon \cdot h$. Thus the offline cost is at most $(2+\epsilon) h$.

Let us now count the online cost during Step 2. By lemma~\ref{lem:step2_count_runs}, there are $N -2= \frac{2}{\epsilon} -2$ executions of $S$ such that WFA has $h$ servers in $R$. By lemma~\ref{lem:cost-step2-h} we get that during each one of those executions, the online cost is $ \epsilon \cdot h^2$. For the rest of Step 2, the WFA incurs a cost of at least $h$, as it moves $h$ servers from $L$ to $R$. We get that overall, during Step 2, the cost of WFA is at least $$ (\frac{2}{\epsilon} -2) \cdot \epsilon \cdot h^2 +h = (2-2\epsilon)h^2 +h .$$ 
\end{proof}

\subsection{Lower bound for WFA on the line}\label{sec:lb-wfa-ln}
The lower bound of section~\ref{sec:lb-wfa-hst} for the WFA can also be applied on the line. The lower bound strategy is the same as the one described for depth-3 HSTs, we just need to replace subtrees by line segments.

 More precisely, the lower bound strategy is applied in an interval $I$ of length 1. Let $L$ and $R$ be the leftmost and rightmost regions of I, of length $\epsilon/2$, for $\epsilon \ll 1$. Similarly, $L_1,L_2$ and $R_1,R_2$ are smaller regions inside $L$ and $R$ respectively. Again, the distance between $L_1$ and $L_2$ ($R_1$ and $R_2$ resp.) is much larger than their length. Whenever the adversary moves to such a region, it places its servers in $h$ equally spaced points inside it. 
Similarly to the proof of Theorem~\ref{thm:lb-wfa}, we can get a lower bound
$h + 1/3$ on the competitive ratio of the WFA on the line when $k=2h$.

Interestingly, the lower bound obtained by this construction for the line has a matching upper bound. As it was observed in~\cite{BEJKP15}, results from~\cite{BK04},~\cite{Kou99} imply that for the line the cost of WFA with $k$ servers WFA$_k$ is at most $(h+1) \mbox{OPT}_h - \mbox{OPT}_k + \mbox{const }$, where $\mbox{OPT}_i$ is the optimal cost using $i$ servers. 
Briefly, this upper bound holds for the following reason: In ~\cite{Kou99} it was shown that in general metrics, WFA$_k \leq 2h \mbox{OPT}_h - \mbox{OPT}_k + \mbox{const} $. However, if we restrict our attention to the line, using the result of~\cite{BK04}, we can get that 
\begin{equation}
\label{eq:wfa_ub_line}
\mbox{WFA}_k \leq (h+1) \mbox{OPT}_h - \mbox{OPT}_k + \mbox{const    }.
\end{equation}   
 See~\cite{BEJKP15} for more details.

In theorem~\ref{thm:lb-wfa} we showed a lower bound $(h+1/3) \mbox{OPT}_h $ on WFA$_k$. We now show that this construction matches the upper bound of~\eqref{eq:wfa_ub_line}. In particular, it suffices to show that $ (h+1/3) \mbox{OPT}_h $ goes arbitrary close to $(h+1) \mbox{OPT}_h - \mbox{OPT}_k  $, or equivalently 

\begin{equation}
\label{eq:opt_hk}
\frac{2\mbox{OPT}_h}{3} \rightarrow \mbox{OPT}_k
\end{equation}

As we showed in the proof of theorem~\ref{thm:lb-wfa}, for every phase of the lower bound strategy, $\mbox{OPT}_h \leq (3+\epsilon)\cdot h $. Moreover it is clear that $\mbox{OPT}_k = 2h$; during a phase the minimum work function value using $k$ servers increases by exactly $2h$. We get that 

\[ \frac{2\mbox{OPT}_h}{3} \leq \frac{2 \cdot(3+\epsilon)\cdot h}{3} = 2h \frac{3+\epsilon}{3} \rightarrow 2h = \mbox{OPT}_k . \]